\documentclass{article}
\usepackage{spconf}
\usepackage{graphicx}
\usepackage{cite}

\usepackage[cmex10]{amsmath}
\usepackage{amssymb}
\usepackage{color}
\usepackage{amsthm}
\usepackage[none]{hyphenat} 

\theoremstyle{mytheoremstyle}
\newtheorem{theorem}{Theorem}

\title{MRT-based Joint Unicast and Multigroup Multicast Transmission in Massive MIMO Systems }
\name{Meysam Sadeghi$^{*}$\thanks{This work was supported in part by the Swedish Research Council (VR), the Swedish Foundation for Strategic Research (SSF), and ELLIIT. Also, it was supported by A*Star SERC project number 142-02-00043 and NSFC 61750110529.}, Emil Bj\"{o}rnson$^{\dagger}$, Erik~G.~Larsson$^{\dagger}$, Chau~Yuen$^{*}$, and~Thomas~L.~Marzetta$^{\ddagger}$}

\address{$^*$ Singapore University of Technology and Design, Singapore \\
$^\dagger$ Department of Electrical Engineering (ISY), Link\"{o}ping University, Link\"{o}ping, Sweden\\
$^\ddagger$ Department of Electrical and Computer Engineering, New York University,  New York, USA \\
Email: meysam@mymail.sutd.edu.sg, emil.bjornson@liu.se, erik.g.larsson@liu.se,\\ yuenchau@sutd.edu.sg, and tom.marzetta@nyu.edu}

\begin{document}
\ninept
\maketitle

\begin{abstract}
We study joint unicast and multigroup multicast transmission in single-cell massive multiple-input-multiple-output (MIMO) systems, under maximum ratio transmission. For the unicast transmission, the objective is to maximize the weighted sum spectral efficiency (SE) of the unicast user terminals (UTs) and for the multicast transmission the objective is to maximize the minimum SE of the multicast UTs. These two problems are coupled to each other in a conflicting manner, due to their shared power resource and interference. To address this, we formulate a multiobjective optimization problem (MOOP). We derive the Pareto boundary of the MOOP analytically and determine the values of the system parameters to achieve any desired Pareto optimal point. Moreover, we prove that the Pareto region is convex, hence the system should serve the unicast and multicast UTs at the same time-frequency resource.
\end{abstract}
\vspace{-0.5em}
\section{Introduction}
\vspace{-0.5em}
In this paper, we study the joint unicast and multigroup multicast transmission in massive multiple-input-multiple-output (MIMO) systems under maximum ratio transmission (MRT). We want to maximize the weighted sum spectral efficiency (SE) of unicast user terminals (UTs) and, at the same time, we want to maximize the minimum SE of multicast UTs, i.e., the so-called max-min fairness (MMF) problem.

A first seminal treatment of multicast transmission is presented in \cite{sidiropoulos2006transmit}, where it is proved that the problem is NP-hard and a semidefinite relaxation (SDR) based method is proposed. The extension of this problem to multigroup multicasting is then studied in \cite{karipidis2008quality} and \cite{christopoulos2014weighted}, where they also use the SDR technique. However, SDR-based techniques suffer from high computational complexity. Considering a single-cell system with an $N$ antenna base station (BS), the SDR-based methods have a complexity of $\mathcal{O}(N^{6.5})$ \cite{karipidis2008quality}, which is prohibitive for massive MIMO systems.

Due to high energy and spectral efficiencies of massive MIMO systems \cite{hoydis2013massive,ngo2013energy,LSAPowerNorm,sadeghiGC14austin}, there has been an increasing interest in presenting multicasting algorithms tailored for massive MIMO \cite{tran2014conic,MeysamMultiComplexity,sadeghi2015multi,Zhengzheng2014,GC2017,YangMulticat,MeysamMMFTWC}. To reduce the computational complexity of SDR-based multicasting methods, successive convex approximation is used in \cite{tran2014conic,MeysamMultiComplexity}, while assuming perfect channel state information (CSI) at the BS and UTs. To jointly address the imperfect CSI and the high computational complexity of massive MIMO multicasting, two different approaches are presented in \cite{Zhengzheng2014} and \cite{YangMulticat}. The first approach exploits the asymptotic orthogonality of the channels in massive MIMO systems to simplify the MMF problem \cite{Zhengzheng2014}. However, it requires an extremely large number of antennas for a reasonable performance, e.g., $N>1000$ \cite{sadeghi2015multi,GC2017}. The other approach, \cite{YangMulticat,GC2017,MeysamMMFTWC}, uses the statistical properties of massive MIMO systems and present practical and computationally efficient massive MIMO multicasting methods with reasonable number of BS antennas, e.g., $N > 50$ \cite{YangMulticat,MeysamMMFTWC}. 

The aforementioned results are for pure multicast transmission. But a practical system should be able to simultaneously serve both unicast and multicast UTs. This has motivated the study of joint unicast and multicast transmission, e.g., \cite{R1,R2,R3}. However, these works neither study massive MIMO systems nor account for imperfect CSI acquisition. The first step in this direction is presented in \cite{JBB}, where a joint beamforming and broadcasting technique for massive MIMO system is proposed, however, its extension to multigroup multicasting is an open problem.

To the best of our knowledge, the coexistence of unicast and multigroup multicast transmission in massive MIMO systems, while accounting for imperfect CSI and the computational complexity, has not been studied in the literature. Studying such a system is challenging because for multicast UTs usually the desired objective is to maximize their minimum SE \cite{tran2014conic,MeysamMultiComplexity,sadeghi2015multi,YangMulticat,GC2017,MeysamMMFTWC}, while for the unicast UTs it is usually to maximize their sum~SE~(SSE) \cite{SumRateMaxPowerAllTWC,SumRateTWC,SumPower}. Noting that we have a shared power, inter-user interference, and two conflicting objectives, one challenging aspect of such a problem is to find a rigorous definition of optimality. Moreover, for a given optimality measure, how should the resources be allocated to reach optimality? In this paper, we address these problems, while our system model accounts for channel estimation and pilot contamination.

\vspace{-.9em}

\section{System Model}
\vspace{-.5em}
We consider a single BS with $N$ antennas that jointly serves $U$ single-antenna unicast UTs and $G$ multicasting groups, where the $g$th multicasting group has $K_{g}$ single-antenna multicast UTs. We assume a UT in the system is either a unicast UT or is a multicast UT. Therefore, the total number of UTs is $U + \sum_{g=1}^{G} K_{g}$, which are arbitrarily distributed in the system. 

We consider a block fading channel model where the channels are static within a coherence interval of $T$ symbols, where $T = C_{B} C_{T}$ with $C_{B}$ and $C_{T}$ being the coherence bandwidth and coherence time, respectively. We denote the channel response between the unicast UT $u$ and the BS as $\mathbf{f}_{u}$ and the channel response between the multicast UT $k$ in group $g$ and the BS as $\mathbf{g}_{gk}$. We consider uncorrelated Rayleigh fading channel responses for UTs, i.e., $\mathbf{f}_{u} \sim \mathcal{CN}(\mathbf{0},\beta_{u} \mathbf{I}_{N})$ and $\mathbf{g}_{gk} \sim \mathcal{CN}(\mathbf{0},\eta_{gk} \mathbf{I}_{N})$, where $\beta_{u}$ and $\eta_{gk}$ are the large-scale fading coefficients for the unicast and multicast UTs, respectively. Note that the measurement campaigns have shown that the achievable rates under independent Rayleigh fading model well match the practically achievable rates \cite{GaoMeasurements}, and trends similar to those seen in independent Rayleigh channels are verified for the measured channels \cite{GaoMeasurements}.

The BS obtains CSI from uplink training using a time division duplexing protocol. In each coherence interval, we have two phases: Uplink pilot transmission and downlink data transmission. During uplink transmission, the UTs send uplink pilots, which enables the BS to estimate the channels. The BS uses these channel estimates to perform downlink precoding. As the coherence interval is limited and due to the large number of multicast UTs, we cannot assign a dedicated orthogonal pilot to each multicast UT in the system, since this will exhaust the resources. Therefore, for multicast UTs, we use the novel co-pilot assignment strategy proposed in \cite{YangMulticat}. This strategy assigns a shared pilot to all the UTs in each multicasting group. Hence, the BS only requires $U+G$ pilots (rather than $U + \sum_{j=1}^{G} K_{j}$ pilots) to simultaneously serve the $U$ unicast UTs and also all the multicast UTs in the system. 

Following the proposed scheme in \cite{YangMulticat}, the BS can estimate $\mathbf{f}_{m}$ by an MMSE estimation as follows:
\begin{align}
\hat{\mathbf{f}}_{u} = \dfrac{\sqrt{\tau p_{u}^{up}} \beta_{u}}{1 + \tau p_{u}^{up} \beta_{u}} \left( \sqrt{\tau p_{u}^{up}} \; \mathbf{f}_{u} + \mathbf{n}_{u} \right) 
\end{align}
with $\mathbf{n}_{u} \sim \mathcal{CN}(\mathbf{0},\mathbf{I}_{N})$ being additive noise and $p_{u}^{up}$ is proportional to the uplink pilot power that has been used by unicast UT $u$. Therefore, $\hat{\mathbf{f}}_{u} \sim \mathcal{CN}(\mathbf{0},\vartheta_{u} \mathbf{I}_{N})$ with $\vartheta_{u} = \tau p_{u}^{up} \beta_{u}^{2}/ \left(1 + \tau p_{u}^{up} \beta_{u} \right)$. We stack the estimated channel vectors between the BS and the $U$ unicast UTs in an $N \times U$ matrix $\hat{\mathbf{F}}= [\hat{\mathbf{f}}_{1}, \ldots,\hat{\mathbf{f}}_{U}]$. The MMSE estimate of $\mathbf{g}_{gk}$ is given as 
\begin{align}
\label{est.mu}
\hat{\mathbf{g}}_{gk} = \dfrac{\sqrt{\tau q_{gk}^{up}} \; \eta_{gk}}{1 + \sum_{t=1}^{K_{g}} \tau q_{gt}^{up} \eta_{gt}}  \left( \sum_{t=1}^{K_{g}} \sqrt{ \tau q_{gt}^{up} } \; \mathbf{g}_{gt} + \mathbf{n}_{g}  \right)
\end{align}
where $\mathbf{n}_{g} \sim \mathcal{CN}(\mathbf{0},\mathbf{I}_{N})$ is the additive noise and $q_{gk}^{up}$ is proportional to the uplink pilot power that has been used by multicast UT $k$ of multicasting group $g$. Therefore $\hat{\mathbf{g}}_{gk} \! \sim \mathcal{CN}(\mathbf{0},\xi_{gk} \mathbf{I}_{N})$ with $\xi_{gk} \!\! = \!\!  \tau q_{gk}^{up} \; \eta_{gk}^{2}/  \left(1 \!+\! \sum_{t=1}^{K_{g}} \tau q_{gt}^{up} \eta_{gt} \right)$. Given \eqref{est.mu}, it is obvious that for every multicast group $g$ the channels estimates of the multicast UTs are equivalent up to a scalar coefficient. We can also estimate the composite channel of all multicast UTs within group $g$ as $\mathbf{g}_{g} = \sum_{t=1}^{K_{g}} \sqrt{ \tau q_{gt}^{up} } \; \mathbf{g}_{gt}$ and we have
\begin{align}
\label{ghat}
\hat{\mathbf{g}}_{g} =  \dfrac{ \sum_{t=1}^{K_{g}} \tau q_{gt}^{up} \eta_{gt}}{1 + \sum_{t=1}^{K_{g}} \tau q_{gt}^{up} \eta_{gt}}   \left( \sum_{t=1}^{K_{g}} \sqrt{ \tau q_{gt}^{up} } \; \mathbf{g}_{gt} + \mathbf{n}_{g} \right)
\end{align}
where $\hat{\mathbf{g}}_{g} \! \sim \mathcal{CN}(\mathbf{0}, \gamma_{g} \mathbf{I}_{N})$ with $\gamma_{g} \!=\! \dfrac{ \left( \sum_{t=1}^{K_{g}} \tau q_{gt}^{up} \eta_{gt} \right)^{2} }{1 \!+\! \sum_{t=1}^{K_{g}} \tau q_{gt}^{up} \eta_{gt}} $. We stack the $G$ composite channel vectors in an $N \times G$ matrix $\hat{\mathbf{G}} = [\hat{\mathbf{g}}_{1},\ldots,\hat{\mathbf{g}}_{G}]$. Note that $\hat{\mathbf{g}}_{gk}$ and $\hat{\mathbf{g}}_{g}$ are equal (to within a multiplicative constant) as follows:
\begin{align}
\label{ghat.to.ghatjk}
\hat{\mathbf{g}}_{gk} = \dfrac{\sqrt{\tau q_{gk}^{up}} \; \eta_{gk}}{\sum_{t=1}^{K_{g}} \tau q_{gt}^{up} \eta_{gt}} \hat{\mathbf{g}}_{g}.
\end{align}

\subsection{Downlink Transmission}
Let us denote the data symbols for the $U$ unicast UTs as $\mathbf{x}=[x_{1},\ldots,x_{U}]^{T}$, where $\mathbf{x} \sim \mathcal{CN}(\mathbf{0},\mathbf{I}_{U})$. We denote the data symbols for the $G$ multicasting groups as $\mathbf{s}=[s_{1},\ldots,s_{G}]^{T}$, where $\mathbf{s} \sim \mathcal{CN}(\mathbf{0},\mathbf{I}_{G})$. Moreover, we assume $\mathbf{x}$ and $\mathbf{s}$ are independent. Then the signal received by $m$th unicast UT is
\begin{align}
\label{unicast-RXsignal}
y_{m} = \mathbf{f}_{m}^{H} \left( \mathbf{V} \mathbf{x} + \mathbf{W} \mathbf{s} \right) + n_{m}
\end{align}
where $n_{m} \sim \mathcal{CN}(0,1)$ is the normalized additive noise, $\mathbf{V} = [\mathbf{v}_{1},\ldots,\mathbf{v}_{U}]$ is the $N \times U$ unicast precoding matrix with $\mathbf{v}_{m}$ being the precoding vector of $m$th unicast UT, and $\mathbf{W}=[\mathbf{w}_{1},\ldots,\mathbf{w}_{G}]$ is the $N \times G$ multicast precoding matrix with $\mathbf{w}_{j}$ being the precoding vector of $j$th multicasting group. The received signal by $k$th multicast UT in $j$th multicasting group is
\begin{align}
\label{multicast-RXsignal}
z_{jk} = \mathbf{g}_{jk}^{H} \left( \mathbf{W} \mathbf{s} + \mathbf{V} \mathbf{x} \right) + n_{jk}
\end{align}
where $n_{jk} \sim \mathcal{CN}(0,1)$ is the normalized additive noise. Therefore the precoding matrix becomes an $N \times (U+G)$ matrix $[\mathbf{V},\mathbf{W}]$. We will detail the structure of $\mathbf{V}$ and $\mathbf{W}$ next.

\vspace{-0.5em}
\section{Achievable Spectral Efficiencies with MRT}
\vspace{-0.5em}
Using MRT, the precoding vector for $m$th unicast UT is
\begin{align}
\label{MRT-unicast}
\mathbf{v}_{m}^{\mathrm{MRT}}=\sqrt{\frac{p_{m}^{dl}}{N \vartheta_{m}}} \;\; \hat{\mathbf{f}}_{m}  
\end{align}
where $p_{m}^{dl}$ is the downlink power of the unicast precoding vector, i.e., $\mathbb{E}[\Vert \mathbf{v}_{m} \Vert^{2}] = p_{m}^{dl}$. The precoding vector for $j$th multicast group is
\begin{align}
\label{MRT-multicast}
\mathbf{w}_{j}^{\mathrm{MRT}} = \sqrt{\frac{q_{j}^{dl}}{N \gamma_{j}}} \;\; \hat{\mathbf{g}}_{j}
\end{align}
where $q_{j}^{dl}$ is the downlink power of the multicast precoding vector, i.e., $\mathbb{E}[\Vert \mathbf{w}_{j} \Vert^{2}] = q_{j}^{dl}$. We assume the total downlink available power at the BS is equal to $P$. Then the power allocated to the downlink unicast and multicast precoding vectors should meet the following condition:
\begin{align}
\label{TotalPower}
P_{un} + P_{mu} \leq P
\end{align}
where $ P_{un}=\sum_{m=1}^{U} p_{m}^{dl}$ and $ P_{mu} = \sum_{j=1}^{G} q_{j}^{dl}$ are the precoding powers used for unicast and multicast transmissions at the BS, respectively.


Consider the unicast UTs. Starting from \eqref{unicast-RXsignal}, we can write the signal received by $m$th unicast UT as follows
\begin{align}
\label{uatf}
y_{m} = \mathbb{E}[\mathbf{f}_{m}^{H}  \mathbf{v}_{m}] x_{m} + (\mathbf{f}_{m}^{H}  \mathbf{v}_{m} - \mathbb{E}[\mathbf{f}_{m}^{H}  \mathbf{v}_{m}])  x_{m}  \notag\\
+  \sum_{u=1, u \neq m}^{U}  \mathbf{f}_{m}^{H}  \mathbf{v}_{u} x_{u} + \sum_{g=1}^{G} \mathbf{f}_{m}^{H} \mathbf{w}_{g} s_{g}  + n_{m}.
\end{align}
Now by applying a common bounding technique \cite[Sec. 2.3]{marzetta2016fundamentals}, which treats the first term of \eqref{uatf} as the desired signal and the remaining terms as effective noise, the following SE is achievable for $m$th unicast UT
	\begin{align}
	\label{SE.MRT.un}
	\mathrm{SE}_{m,un} &= \left(1 - \dfrac{\tau}{T} \right) \log_{2}(1+\mathrm{SINR}_{m,un})
	\\
	\label{SINR.MRT.un}
	\mathrm{SINR}_{m,un} &= \dfrac{N p_{m}^{dl} \vartheta_{m} }{1  + \beta_{m} ( P_{un} + P_{mu})}.
	\end{align}
Also, the following SE is achievable for multicast UT $k$ of group $j$
	\begin{align}
	\label{SE.MRT.mu}
	\mathrm{SE}_{jk,mu} &= \left(1 - \dfrac{\tau}{T} \right) \log_{2}(1+\mathrm{SINR}_{jk,mu})
	\\
	\label{SINR.MRT.mu}
	\mathrm{SINR}_{jk,mu} &= \dfrac{N  q_{j}^{dl}  \xi_{jk}  }{1  + \eta_{jk} (P_{mu} + P_{un})}.
	\end{align}

\vspace{-1.5em}
\section{Optimal Resource Allocation}
\vspace{-0.5em}
Considering our proposed single cell joint unicast and multi-group multicast system, the MMF problem for the multicast UTs can be stated as
\begin{align}
\label{MMF}
\mathcal{P}1: \underset{\{q_{j}^{dl}\},\{q_{jk}^{up}\}, \tau}{\text{maximize}} \; \underset{j \in \mathcal{G}, k \in \mathcal{K}_{j}}{\text{min}} \;\;& \mathrm{SE}_{jk,mu}
\\
s.t. \; & P_{mu} + P_{un} \leq P  \tag{\ref{MMF}-i}
\\
& 0 \leq q_{j}^{dl}   \tag{\ref{MMF}-ii}
\\
& 0 \leq \tau q_{jk}^{up} \leq E_{jk} \tag{\ref{MMF}-iii}
\\
& \tau \in \{U\!\!+\!G,\ldots,T \} \tag{\ref{MMF}-iv}
\end{align}
where $P$ is the total downlink power at the BS and $P_{un}$ is assumed to be a fixed known quantity. Also $E_{jk}$ is the maximum energy limit of the multicast UT $k$ in group $j$ per pilot transmission. Hereafter we denote an objective value of $\mathcal{P}1$ that is obtained for a set of feasible decision variables of $\mathcal{P}1$ by $O_{mu}$, and we denote its optimal objective value by $O_{mu}^{*}$.

\begin{theorem}
	\label{multicastOpt}
	For a fixed value of $P_{un}$ where $0 \leq P_{un} \leq P$, at the optimal solution of $\mathcal{P}1$ all the multicast UTs will have equal SEs, i.e., $O_{mu}^{*}(P_{un}) = \mathrm{SE}_{jk} \; \forall j,k$, which is the optimal objective value of $\mathcal{P}1$ and it is equal to 
	\begin{align}
	\label{O1}
	&O_{mu}^{*} (P_{un}) = \left( 1 - \dfrac{U\!\!+\!G}{T} \right) \times
	\\
	&\log_{2} \left( 1 + \dfrac{P_{mu} N}{P \sum_{j=1}^{G} K_{j}  + \sum_{j=1}^{G} \frac{1}{\Upsilon_{j}}  + \sum_{j=1}^{G} \sum_{t=1}^{K_{j}} \frac{1}{\eta_{jt}}} \right)   \notag 
	\end{align}
	where $\Upsilon_{j} = \min_{k \in \mathcal{K}_{j}} \frac{E_{jk} \eta_{jk}^{2}}{1+\eta_{jk}P}$ and $P_{mu} = P - P_{un}$. Also the optimal values of decision variables are: $\tau^{*} = U\!\!+\!G$, $q_{jk}^{up*} = \frac{1+ \eta_{jk}P}{(U\!\!+\!G)\eta_{jk}^{2}} \Upsilon_{j}$, $q_{j}^{dl*} = \frac{\Gamma}{N \Upsilon_{j}}  (1 + \sum_{t=1}^{K_{j}} x_{jt}^{*} \eta_{jt})$ with $x_{jk}^{*} = \frac{1+ \eta_{jk}P}{\eta_{jk}^{2}} \Upsilon_{j}$.
\end{theorem}

For unicast UTs we can optimize the sum of the weighted SEs of the unicast UTs, which is a common metric of interest in unicast transmission \cite{SumRateMaxPowerAllTWC,LowComSumRateMassiveMIMO,SumRateTWC,SumPower,SumRateMaxTSP}. Given our proposed single cell joint unicast and multi-group multicast system, the SSE problem for unicast UTs is 
\begin{align}
\label{sumrate}
\mathcal{P}2: \underset{ \{ p_{m}^{dl} \}, \{ p_{m}^{u} \}, \tau }{\text{maximize}}  \;\;&  \sum_{m=1}^{U} \alpha_{m} \mathrm{SE}_{m,un}
\\
s.t. \;\; & P_{un} \leq P - P_{mu} \tag{\ref{sumrate}-i}
\\
& 0 \leq p_{m}^{dl}   \tag{\ref{sumrate}-ii}
\\
& 0 \leq \tau p_{m}^{up} \leq E_{m}      \tag{\ref{sumrate}-iii}
\\
& \tau \in \{U\!\!+\!G,\ldots,T \}        \tag{\ref{sumrate}-iv}
\end{align}
where $E_{m}$ is maximum energy limit of the $m$th unicast UT per pilot transmission and $\alpha_{m}$ is the weight assigned to this user's SE. Hereafter we denote an objective value of $\mathcal{P}2$ that is obtained for a set of feasible decision variables of $\mathcal{P}2$ by $O_{un}$ and we denote its optimal objective value by $O_{un}^{*}$. Now we have the following result.

\begin{theorem}
	\label{unicastOpt}
	For a fixed value of $P_{mu}$ and $0 \leq P_{mu} \leq P$, the optimal solution to problem $\mathcal{P}2$ is $\tau^{*} =U+G$, $p_{m}^{up*} = \frac{E_{m}}{U+G}$, and 
	\begin{align}
	\label{P3Sol}
	p_{m}^{dl*} = \max \left\lbrace  0,\dfrac{\alpha_{m}}{\nu \ln 2 } - \dfrac{1+\beta_{m}P}{N \vartheta_{m}^{*}} \right\rbrace
	\end{align}
	where $\vartheta_{m}^{*}= \frac{E_{m} \beta_{m}^{2}}{1+E_{m} \beta_{m}}$ and $\nu$ is selected to satisfy $P_{un} = P - P_{mu}$. The optimal objective value of $\mathcal{P}2$ becomes
	\begin{align}
	\label{O2}
	O_{un}^{*}(P_{\!mu}) \!=\!  \left(1\! -\! \frac{U\!+\!G}{T}\right)  \sum_{m=1}^{U} \! \alpha_{m} \! \log_{2} \left(1 \!+\! \dfrac{N p_{m}^{dl*} \vartheta_{m}^{*} }{1 \! + \! \beta_{m} P} \right).
	\end{align}
\end{theorem}

\vspace{-0.5em}
\section{The Pareto Region}
\vspace{-0.5em}
Based on Theorems \ref{multicastOpt} and \ref{unicastOpt}, it is obvious that $O_{mu}^{*}(P_{un})$ and $O_{un}^{*}(P_{mu})$ are coupled in a conflicting manner, i.e., $P=P_{un}+P_{mu}$, such that improving $O_{mu}^{*}(P_{un})$ would degrade $O_{un}^{*}(P_{mu})$, and vice versa. Therefore we need to consider a multi-objective optimization framework for the considered system that describes the jointly achievable operating points. This description enables us to balance the two conflicting objectives efficiently. To present such a description we need to state a MOOP and derive its Pareto boundary \cite{bjornson2014multiobjective,marler2004survey}. The MOOP for the proposed joint unicast and multigroup multicast massive MIMO system is given by \cite{marler2004survey}
\begin{align}
\label{MOOP}
\mathcal{M}: \underset{\mathbf{x}}{\text{maximize}} &\quad [O_{mu}(\mathbf{x}), O_{un}(\mathbf{x})]^{T}
\\
s.t. &\quad    \mathbf{x} \in  \mathcal{X}    \tag{\ref{MOOP}-i}
\end{align}
where $\mathbf{x} =  ( \{ q_{j}^{dl} \}, \{ q_{jk}^{up} \},  \{ p_{m}^{dl} \}, \{ p_{m}^{up} \}, \tau )$, and $\mathcal{X}$ is the so called resource bundle \cite{bjornson2014multiobjective}, which is achieved by bundling the feasible space of $\mathcal{P}1$ and $\mathcal{P}2$ as follows:\\
$\mathcal{X} = \{    ( \{ q_{j}^{dl} \}, \{ q_{jk}^{up} \},  \{ p_{m}^{dl} \}, \{ p_{m}^{up} \}, \tau )  \; | \;  0 \leq q_{j}^{dl}, 0 \leq \tau q_{jk}^{up} \leq E_{jk}, 0 \leq p_{m}^{dl}, \; 0 \leq \tau p_{m}^{up} \leq E_{m}, P_{un} + P_{mu} \leq P, \tau \in \{ U+G,\ldots,T\} \}$. 

Note the difference between $O_{mu}(\mathbf{x})$ and $O_{mu}^{*}(P_{un})$. $O_{mu}(\mathbf{x})$ is the objective achieved for $\mathcal{P}1$ given $\mathbf{x} \in \mathcal{X}$, while $O_{mu}^{*}(P_{un})$ is the optimal objective achieved for $\mathcal{P}1$ based on Theorem \ref{multicastOpt}, given the unicast power is fixed to $P_{un}$. Similarly, $O_{un}(\mathbf{x})$ is the objective achieved for $\mathcal{P}2$ given $\mathbf{x} \in \mathcal{X}$, while $O_{un}^{*}(P_{mu})$ is the optimal objective achieved for $\mathcal{P}2$ based on Theorem \ref{unicastOpt}, given the multicast power is fixed to $P_{mu}$.

The so-called \textit{attainable objective set} of the MOOP $\mathcal{M}$, given in \eqref{MOOP}, is \cite{bjornson2014multiobjective}
\begin{align}
\mathcal{S} = \{(O_{mu}(\mathbf{x}), O_{un}(\mathbf{x})) | \mathbf{x} \in \mathcal{X} \}.
\end{align}
The strong Pareto boundary of $\mathcal{M}$ is a set $\mathcal{B}_{s}$ containing all the tuples $ (O_{mu}(\mathbf{x}^{*}), O_{un}(\mathbf{x}^{*}))$ such that $\mathbf{x}^{*} \in \mathcal{X}$ and $\nexists \mathbf{y} \in \mathcal{X}$ for which either $ O_{mu}(\mathbf{x}^{*}) < O_{mu}(\mathbf{y})$ and $O_{un}(\mathbf{x}^{*})\leq O_{un}(\mathbf{y})$, or $ O_{mu}(\mathbf{x}^{*}) \leq O_{mu}(\mathbf{y})$ and $O_{un}(\mathbf{x}^{*}) <  O_{un}(\mathbf{y})$. In this case $\mathbf{x}^{*}$ is called a Pareto optimal point \cite{marler2004survey}. Moreover, a point $\mathbf{z}^{*} \in \mathcal{X}$ is called a weak Pareto optimal point if $\nexists \mathbf{y} \in \mathcal{X}$ such that $O_{mu}(\mathbf{z}^{*}) < O_{mu}(\mathbf{y}) , O_{un}(\mathbf{z}^{*}) < O_{un}(\mathbf{y})$ \cite{marler2004survey}. Note that every strong Pareto optimal point is also a weak Pareto optimal point, but the converse is not true. The set $\mathcal{B}_{w}$ that contains all the tuples $(O_{mu}(\mathbf{z}^{*}), O_{un}(\mathbf{z}^{*}))$ where $\mathbf{z}^{*}$ is a weak Pareto optimal point is called the weak Pareto boundary \cite{bjornson2014multiobjective,marler2004survey}. We have the following theorem for the Pareto boundary and Pareto optimal points of $\mathcal{M}$.

\begin{theorem}
	\label{Theo_ParetoBound}
	The MOOP \eqref{MOOP} of the considered joint unicast and multigroup multicast massive MIMO system, does not have any weak Pareto optimal points and its strong Pareto boundary is analytically described by
	\vspace{-0.5em}
	\begin{align}
	\mathcal{B}_{s} =  \{  (O_{mu}^{*}(P_{un}), O_{un}^{*}(P_{mu})) \; | \;  P_{mu} + P_{un} = P, \notag \\ 
	0 \leq P_{mu} \leq P ,  0 \leq P_{un} \leq P \}.
	\end{align}
	Moreover, $(O_{mu}^{*}(P_{un}), O_{un}^{*}(P_{mu})) \in \mathcal{B}_{s}$ is achieved when \newline $( \{ q_{j}^{dl*} \}, \{ q_{jk}^{up*} \},  \{ p_{m}^{dl*} \}, \{ p_{m}^{up*} \}, \tau^{*} )$ are obtained either from Theorems \ref{multicastOpt} and \ref{unicastOpt}.
\end{theorem}

\vspace{-0.7em}
Theorem \ref{Theo_ParetoBound} is important due to the following reasons. First, it describes all the Pareto optimal points that can be obtained in the considered system. As each Pareto optimal point describes a particular trade-off between the two objectives (SSE and MMF) of the system, it elaborates the set of efficient operating points from which the network designer can select one based on network requirements. Second, Theorem \ref{Theo_ParetoBound} not only describes the Pareto boundary points, but also determines the exact values of the system parameters (uplink training powers, downlink transmission powers, and the pilot length) to achieve such points. So Theorem \ref{Theo_ParetoBound} exactly describes, in a joint unicast and multi-group multicast massive MIMO system, what can be achieved and how it can be achieved.

\begin{theorem}
	\label{convexity}
	$\mathcal{S}$, the attainable set of MOOP $\mathcal{M}$, is convex.
\end{theorem}

\begin{proof}
	Consider any two arbitrary tuples $\big(\! O_{mu}(\mathbf{x}), O_{un}(\mathbf{x}) \! \big) \!\! \in \! \mathcal{S}$ and $\big( O_{mu}(\mathbf{y}),\! O_{un}(\mathbf{y})  \big) \! \!\in \mathcal{S}$. Define
	\begin{align}
	\big( O_{mu}(\mathbf{\alpha}), O_{un}(\mathbf{\alpha})  \big) \!\! =& \alpha \big( O_{mu}(\mathbf{x}), O_{un}(\mathbf{x})  \big)  \\+& (1- \alpha) \big( O_{mu}(\mathbf{y}), O_{un}(\mathbf{y})  \big)  \;\;\; \alpha \in [0,1]. \notag
	\end{align}
	Now we show that $\big( O_{mu}(\mathbf{\alpha}),\! O_{un}(\mathbf{\alpha})  \big)\! \in \! \mathcal{S}$. Define $P_{mu}(\mathbf{x}) = \sum_{j=\!1}^{G}\! q_{j}^{dl} (\mathbf{x})$, $P_{un}\!(\mathbf{x})\! =\! \sum_{m=1}^{U} \! p_{m}^{dl} (\mathbf{x})$, $P_{mu}(\mathbf{y}) \!=\! \sum_{j=\!1}^{G} \!q_{j}^{dl} (\mathbf{y})$, and $P_{un}(\mathbf{y}) = \sum_{m=1}^{U} p_{m}^{dl} (\mathbf{y})$. As $\mathbf{x}, \mathbf{y} \in \mathcal{X}$, $P_{mu}(\mathbf{x}) + P_{un}(\mathbf{x}) \leq P$ and $P_{mu}(\mathbf{y}) + P_{un}(\mathbf{y}) \leq P$. Based on Theorem \ref{multicastOpt}, using $P_{mu}(\alpha) = \alpha P_{mu}(\mathbf{x}) + (1-\alpha)P_{mu}(\mathbf{y})$ for multicast transmission, we have $O^{*}_{mu}\big(P - P_{mu}(\alpha)\big) \geq O_{mu}(\alpha)$. As $P - P_{mu}(\alpha) \geq P_{un}(\alpha)$, based on Theorem \ref{unicastOpt}, using $P - P_{mu}(\alpha)$ for unicast transmission, we have $O^{*}_{un}\big(P_{mu}(\alpha)\big) \geq O_{un}(\alpha)$. Therefore $\forall \alpha \in [0,1]$, $\exists \big(O^{*}_{mu}, O^{*}_{un}\big) \in \mathcal{S} | \big(O^{*}_{mu}, O^{*}_{un}\big) \geq \big(O_{mu}(\alpha), O_{un}(\alpha)\big)$. As \eqref{O1} and \eqref{O2} are continuous, we can obtain $\big(O_{mu}(\alpha), O_{un}(\alpha)\big)$ by reducing $P_{mu}$ and $\{p_{m}^{dl}\}$ in \eqref{O1} and \eqref{O2} (or \eqref{O1} and \eqref{O2}), which completes the proof.
\end{proof}
\vspace{-0.35em}
Theorem \ref{convexity} proves that the Pareto region is convex. Therefore, it shows that any point of the Pareto region can be obtained by spatial multiplexing of unicast and multicast UTs, and there is no need for orthogonal time or frequency sharing.

\vspace{-0.5em}
\section{Numerical Results}
\vspace{-0.5em}
In this section, we use numerical simulations to verify Theorems~\ref{Theo_ParetoBound} and \ref{convexity}. We consider a system with $U=20$, $G=10$, and $K_{g} = 100$ $\forall g$, i.e., we have $1020$ UTs. The cell radius is considered to be $500$ meters and the unicast and multicast UTs are randomly and uniformly distributed in the cell excluding an inner circle of radius $35$ meters. The large-scale fading coefficient for the multicast UT $k$ in group $j$ is modeled as $\beta_{jk}^{m} = \bar{d}/x_{jk}^{\nu}$, where $\nu=3.76$ is the path-loss exponent, $x_{jk}$ is the distance between the UT and the BS, and $\bar{d} = 10^{-3.5}$ is a constant that regulates the channel attenuation at $35$ meters \cite{3GPPmodel}. Similarly, the large-scale fading coefficient for the unicast UT $m$ is modeled as $\beta_{m}=\bar{d}/x_{m}^{\nu}$, where $x_{m}$ is the distance between this UT and the BS.

The transmission bandwidth is assumed to be $W=20$ MHz, the coherence bandwidth and coherence time are considered to be $200$ kHz and $1$ ms, which results in a coherence interval of length $200$ symbols \cite{yang2013performance}. The noise power spectral density is considered to be $\sigma^{2} = -174$ dBm/Hz. The total downlink power $\bar{P}$ is set to $10$ Watts and its corresponding normalized value is $P = \bar{P} / (W \cdot \sigma^{2})$. We assume $E_{m}$ and $E_{jk}$ are equal to $2$ micro-joule $\forall \; m,j,k$.  

Fig.~\ref{Fig1} presents the Pareto region and the Pareto boundary of the MOOP $\mathcal{M}$, for different values of $N$. For each value of $N$, the ratio $P_{un} / P$ is changed from 0 to 1 with steps of $P/20$, which are denoted by $21$ marker points. First, note that the Pareto region for all the considered values of $N$ is convex, as shown in Theorem~\ref{convexity}. Second, the Pareto boundary does not have any weak Pareto optimal points, refer to the Theorem~\ref{Theo_ParetoBound}. Third, as illustrated by the radial lines (the red, black, and blue radial lines which are presenting a specific ratio between $P_{un}$ and $P_{mu}$), adding more antennas improves the objective of both unicast and multicast transmissions. This is due to the improved coherent transmission of signals and it is obtained by employing a large-scale antenna array.

\begin{figure}[]
	\centering
	\includegraphics[width=.99\columnwidth, trim={3cm 6.5cm 1.8cm 6.5cm},clip]{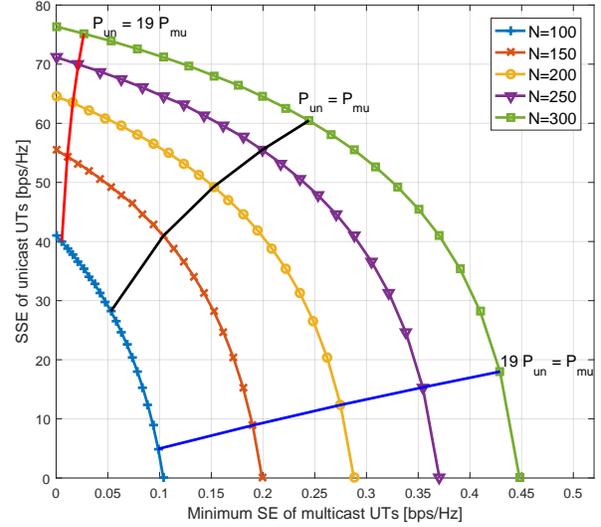}
	\caption{Pareto Region of MOOP $\mathcal{M}$.}
	\label{Fig1}
	\vspace{-0.5em}
\end{figure}

\vspace{-0.5em}
\section{Conclusion}
\vspace{-0.5em}
We studied joint unicast and multigroup multicast transmission in massive MIMO systems.  In these systems, as we have two conflicting objectives, it is challenging to define a notion of optimality. Therefore, we presented a multiobjective optimization framework and proposed a definition of optimality for these systems based on their Pareto region. We expressed the Pareto boundary analytically and proved the Pareto region is convex. Therefore, these systems should use spatial multiplexing to serve both unicast and multicast UTs simultaneously.


\bibliographystyle{IEEEtran}
\bibliography{IEEEabrv,ICASSPJoint}

\end{document}